\documentclass[12pt]{article}
\usepackage[utf8]{inputenc}
\usepackage[T1]{fontenc}
\usepackage{array}
\usepackage{amsmath, amsthm}
\usepackage{amssymb}
\usepackage{refcheck}
\usepackage{enumitem}

\newtheorem{theorem}{Theorem}[section]
\newtheorem{proposition}[theorem]{Proposition}

\newtheorem{lemma}[theorem]{Lemma}

\newtheorem{definition}[theorem]{Definition}

\newtheorem{remark}[theorem]{Remark}

\newcommand{\Factor}{F}
\newcommand{\pfl}{L}
\newcommand{\resalpha}{5}
\DeclareMathOperator{\rec}{rec} 
\DeclareMathOperator{\Prefix}{Prf}
\DeclareMathOperator{\Suffix}{Suf}
\DeclareMathOperator{\occur}{occur}

\begin{document}

\title{Construction of a bi-infinite power free word with a given factor and a non-recurrent letter}

\author{Josef Rukavicka\thanks{Department of Mathematics,
Faculty of Nuclear Sciences and Physical Engineering, Czech Technical University in Prague
(josef.rukavicka@seznam.cz).}}

\date{\small{April 16, 2022}\\
   \small Mathematics Subject Classification: 68R15}

\maketitle

\begin{abstract}
Let $L_{k,\alpha}^{\mathbb{Z}}$ denote the set of all bi-infinite $\alpha$-power free words over an alphabet with $k$ letters, where $\alpha$ is a positive rational number and $k$ is positive integer.
We prove that if $\alpha\geq \resalpha$, $k\geq 3$, $v\in L_{k,\alpha}^{\mathbb{Z}}$, and $w$ is a finite factor of $v$, then there are $\widetilde v\in L_{k,\alpha}^{\mathbb{Z}}$ and a letter $x$ such that $w$ is a factor of $\widetilde v$ and $x$ has only a finitely many occurrences in $\widetilde v$.
\\
\\
\noindent
\textbf{Keywords:} Power Free Words; Extension property; Recurrent factor
\end{abstract}

\section{Introduction}
An $\alpha$-\emph{power} of a nonempty word $r$ is the word $r^{\alpha}=rr\cdots rt$ such that $\frac{\vert r^{\alpha}\vert}{\vert r\vert}=\alpha$ and $t$ is a prefix of $r$, where $\alpha>0$ is a rational number. For example $(1234)^{3}=123412341234$ and $(1234)^{\frac{7}{4}}=1234123$.

Suppose a finite or infinite word $w$. Let \[\begin{split}\Theta(w)=\{(r,\alpha)\mid r^{\alpha}\mbox{ is a factor of }w \mbox{ and }r\mbox{ is a nonempty word and }\\ \alpha\mbox{ is a positive rational number}\}\mbox{.}\end{split}\]
We say that $w$ is $\alpha$-\emph{power free} if \[\{(r,\beta)\in\Theta(w)\mid \beta\geq \alpha\}=\emptyset\] 
and we say that $w$ is $\alpha^+$-power free if \[\{(r,\beta)\in\Theta(w)\mid \beta>\alpha\}=\emptyset\mbox{.}\]  
The power free words include well known square free ($2$-power free), overlap free ($2^+$-power free), and cube free words ($3$-power free). In \cite{Rampersad_Narad2007} and \cite{10.1007/978-3-642-22321-1_3}, the reader can find some surveys on the topic of power free words.


In 1985, Restivo and Salemi presented a list of five problems concerning the extendability of power free words \cite{10.1007/978-3-642-82456-2_20}. For the current article,  Problem $4$ and Problem $5$ are of interest:
\begin{itemize} 
\item
Problem $4$:  Given finite $\alpha$-power-free words $u$ and $v$, decide whether there is a transition word $w$, such that $uwu$ is $\alpha$-power free.
\item
Problem $5$: Given finite $\alpha$-power-free words $u$ and $v$, find a transition word $w$, if it exists.
\end{itemize}
In \cite{10.1007/978-3-030-19955-5_27}, a recent survey on the solution of all the five problems can be found. In particular, the problems $4$ and $5$ are solved for some overlap free binary words. In addition, in \cite{10.1007/978-3-030-19955-5_27} the authors prove that: For every pair $(u,v)$ of cube free words over an alphabet with $k$ letters, if $u$ can be infinitely extended to the right and $v$ can be infinitely extended to the left respecting the cube-freeness property, then there exists a “transition” word $w$ over the same alphabet such that $uwv$ is cube free.

Let $\mathbb{N}$ denote the set of positive integers and let $\mathbb{Q}$ denote the set of rational numbers.

\begin{definition}(see \cite[Definition $1$]{10.1007/978-3-030-48516-0_22})
\[
\begin{split}
\Upsilon=
\{(k,\alpha)\mid k\in \mathbb{N}\mbox{ and }\alpha\in \mathbb{Q}\mbox{ and }k=3 \mbox{ and }\alpha>2\}\\ \cup\{(k,\alpha)\mid k\in \mathbb{N}\mbox{ and }\alpha\in \mathbb{Q}\mbox{ and }k>3\mbox{ and }\alpha\geq 2\}\\ \cup\{(k,\alpha^+)\mid k\in \mathbb{N}\mbox{ and }\alpha\in \mathbb{Q}\mbox{ and } k\geq3\mbox{ and }\alpha\geq 2\}\mbox{.}
\end{split}
\]
\end{definition}
\begin{remark}(see \cite[Remark $1$]{10.1007/978-3-030-48516-0_22})
The definition of $\Upsilon$ says that:
If $(k,\alpha)\in \Upsilon$ and $\alpha$ is a ``number with $+$'' then $k\geq 3$ and $\alpha\geq 2$.
If $(k,\alpha)\in \Upsilon$ and $\alpha$ is ``just'' a number then $k=3$ and $\alpha>2$ or $k>3$ and $\alpha\geq 2$.
\end{remark}

In \cite{10.1007/978-3-030-48516-0_22}, it was shown that if  $(k,\alpha)\in\Upsilon$ then we have that:
For every pair $(u,v)$ of $\alpha$-power free words over an alphabet with $k$ letters, if $u$ can be infinitely extended to the right and $v$ can be infinitely extended to the left respecting the $\alpha$-freeness property, then there exists a “transition” word $w$ over the same alphabet such that $uwv$ is $\alpha$-power free. Also it was shown how to construct the word $w$. Less formally said, the results from \cite{10.1007/978-3-030-48516-0_22} solve Problem $4$ and Problem $5$ for a wide variety of power free languages. 

To prove the results in \cite{10.1007/978-3-030-48516-0_22}, the author showed that: If $v$ is a right (left) infinite $\alpha$-power word with a factor $w$ and $x$ is a letter, then there is a right (left) infinite $\alpha$-power free word $\widetilde v$ such that $\widetilde v$ contains $v$ as a factor and $x$ is not recurrent in $\widetilde v$. In other words, $\widetilde v$ has only a finite number of occurrences of the letter $x$. Also the construction of $\widetilde v$ has been presented. The infinite $\alpha$-power free words with the non-recurrent letter $x$ have then been used to construct then transition words. The results were shown for $\alpha$-power free words over an alphabet with $k$ letters, where $(k,\alpha)\in\Upsilon$.

Let 
\[\widetilde\Upsilon=
\{(k,\alpha)\mid k\in \mathbb{N}\mbox{ and }\alpha\in \mathbb{Q}\mbox{ and }k\geq 3 \mbox{ and }\alpha\geq \resalpha\}\mbox{.}\]

In the current article, we generalize the construction of right and left infinite power free words with a non-recurrent letter to bi-infinite power free words. We prove our result for $\alpha$-power free words over an alphabet with $k$ letter, where $(k,\alpha)\in\widetilde \Upsilon$. Note that $\widetilde \Upsilon\subset \Upsilon$. $L_{k,\alpha}^{\mathbb{Z}}$ denote the set of all bi-infinite $\alpha$-power free words over an alphabet with $k$ letters, where $\alpha$ is a positive rational number and $k$ is positive integer. Formally, our main theorem states that if $(k,\alpha)\in\widetilde \Upsilon$, $v\in L_{k,\alpha}^{\mathbb{Z}}$, and $w$ is a finite factor of $v$, then there are $\widetilde v\in L_{k,\alpha}^{\mathbb{Z}}$ and a letter $x$ such that $w$ is a factor of $\widetilde v$ and $x$ has only a finitely many occurrences in $\widetilde v$.

To prove the current result, we apply the ideas from \cite{10.1007/978-3-030-48516-0_22}. 
Let $v_1$ be left infinite $\alpha$-power word with a non-recurrent letter $x$ and let $v_2$ be a right infinite $\alpha$-power word with no occurrence of the letter $x$. By means of some elaborated observations of recurrent factors, we identify a suffix $\widetilde v_1$ of $v_1$ such that if $\widehat v_1$ is a left infinite word with $v_1=\widehat v_1\widetilde v_1$, then $\widehat v_1v_2$ is a bi-infinite $\alpha$-power free word. 

\section{Preliminaries}
Let $\Sigma_k$ denote an alphabet with $k$ letters. 
Let $\Sigma_k^+$ denote the set of all nonempty finite words over $\Sigma_k$, let $\epsilon$ denote the empty word, let $\Sigma_k^*=\Sigma_k^+\cup\{\epsilon\}$, let $\Sigma_k^{\mathbb{N},R}$ denote the set of all right infinite words over $\Sigma_k$, let $\Sigma_k^{\mathbb{N},L}$ denote the set of all left infinite words over $\Sigma_k$, and let $\Sigma_k^{\mathbb{Z}}$ denote the set of all bi-infinite words over $\Sigma_k$.

Let $\Sigma_k^{\infty}=\Sigma_k^{\mathbb{N},L}\cup \Sigma_k^{\mathbb{N},R}\cup \Sigma_k^{\mathbb{Z}}$. We call $w\in \Sigma_k^{\infty}$ an infinite word.

Let $\occur(w,t)$ denote the number of all occurrences of the nonempty factor $t\in \Sigma_k^+$ in the word $w\in \Sigma_k^*\cup\Sigma_k^{\infty}$. If $w\in \Sigma_k^{\infty}$ and $\occur(w,t)=\infty$, then we call $t$ a \emph{recurrent} factor in $w$.

Let $\Factor(w)$ denote the set of all finite factors of a finite or infinite word $w\in \Sigma_k^*\cup\Sigma_k^{\infty}$. The set $\Factor(w)$ contains the empty word and if $w$ is finite then also $w\in \Factor(w)$. Let $\Factor_{\rec}(w)\subseteq \Factor(w)$ denote the set of all recurrent nonempty factors of $w\in \Sigma_k^{\infty}$.

Let $\Prefix(w)\subseteq \Factor(w)$ denote the set of all prefixes of $w\in \Sigma_k^*\cup\Sigma_k^{\mathbb{N},R}$ and let $\Suffix(w)\subseteq \Factor(w)$ denote the set of all suffixes of $w\in \Sigma_k^*\cup\Sigma_k^{\mathbb{N},L}$. We define that $\epsilon\in \Prefix(w)\cap\Suffix(w)$ and if $w$ is finite then also $w\in \Prefix(w)\cap\Suffix(w)$. 

Let $L_{k,\alpha}$ denote an $\alpha$-power free language over an alphabet with $k$ letters, where $\alpha$ is a positive rational number and $k$ is positive integer.
We have that $\pfl_{k,\alpha}=\{w\in\Sigma_k^*\mid w \mbox{ is }\alpha\mbox{-power free}\}$. Let \[\pfl_{k,\alpha}^{\infty}\subseteq\Sigma_k^{\infty}=\{w\in \Sigma_k^{\infty}\mid\Factor(w)\subseteq \pfl_{k,\alpha}\}\mbox{.}\] Thus $\pfl_{k,\alpha}^{\infty}$ denotes the set of all infinite $\alpha$-power free words over $\Sigma_k$. In addition we define $\pfl_{k,\alpha}^{\mathbb{N},R}=\pfl_{k,\alpha}^{\infty}\cap\Sigma_k^{\mathbb{N},R}$, $\pfl_{k,\alpha}^{\mathbb{N},L}=\pfl_{k,\alpha}^{\infty}\cap\Sigma_k^{\mathbb{N},L}$, and $\pfl_{k,\alpha}^{\mathbb{Z}}=\pfl_{k,\alpha}^{\infty}\cap\Sigma_k^{\mathbb{Z}}$; it means the sets of right infinite, left infinite, and bi-infinite $\alpha$-power free words, respectively.

Let $\mathbb{Z}$ denote the set of integers.
We define the \emph{reverse} $w^R$ of a finite or infinite word $w=\Sigma_k^*\cup\Sigma_k^{\infty}$ as follows: \begin{itemize}
\item $\epsilon^R=\epsilon$.
\item If $w\in\Sigma_k^+$ and $w=w_1w_2\cdots w_m$, where $w_i\in \Sigma_k$ and $1\leq i\leq m$, then $w^R=w_m\cdots w_2w_1$. 
\item If $w\in \Sigma_k^{\mathbb{N},L}$ and $w=\cdots w_2w_1$, where $w_i\in \Sigma_k$ and $i\in \mathbb{N}$, then $w^R=w_1w_2\cdots\in \Sigma_k^{\mathbb{N},R}$. \item If $w\in \Sigma_k^{\mathbb{N},R}$ and $w=w_1w_2\cdots$, where $w_i\in \Sigma_k$ and $i\in \mathbb{N}$, then $w^R=\cdots w_2w_1\in \Sigma_k^{\mathbb{N},L}$. 
\item If $w\in\Sigma_k^{\mathbb{Z}}$ and $w=\cdots w_{-2}w_{-1}w_{0}w_{1}w_{2}\cdots$, where $w_{i}\in\Sigma_k$ and $i\in\mathbb{Z}$, then $w^R=\cdots w_2w_1w_0w_{-1}w_{-2}\cdots$.
\end{itemize}

\begin{remark}
It is obvious that the reverse function preserves the power-freeness and that every factor of an $\alpha$-power free word is also $\alpha$-power free.
\end{remark}

The next proposition is a ``reformulation'' of Corollary $1$ from \cite{10.1007/978-3-030-48516-0_22} using only the notation of the current article. 
\begin{proposition}(reformulation of \cite[Corollary $1$]{10.1007/978-3-030-48516-0_22})
\label{nb3rdyd887j}
If $(k,\alpha)\in\Upsilon$, $v\in\pfl_{k,\alpha}^{\mathbb{N},L}$, $z\in\Suffix(v)$, $x\in\Factor_{\rec}(v)\cap\Sigma_k$, $s\in\pfl_{k,\alpha}^{\mathbb{N},L}$, and $x\not\in\Factor(s)$, then there is a finite word $u\in\Sigma_k^*$ such that  $z\in\Suffix(su)$ and $su\in\pfl_{k,\alpha}^{\mathbb{N},L}$.
\end{proposition}
\begin{remark}
Proposition \ref{nb3rdyd887j} says that if $z$ a finite power free word that can be extended to left infinite power free word having a letter $x$ as a recurrent factor and $s$ is a left infinite power free word not containing the letter $x$ as a factor, then there is a left infinite power free word containing $z$ as a suffix and having only a finite number of occurrences of $x$.
\end{remark}

The following elementary lemma was shown in \cite{10.1007/978-3-030-48516-0_22}. For the reader's convenience we present the lemma with the proof.
\begin{lemma}(reformulation \cite[Lemma $2$]{10.1007/978-3-030-48516-0_22})
\label{dy77ejhfiffu}
If $k\geq 3$ and $\alpha>2$ then $\pfl_{k-1,\alpha}^{\mathbb{N},R}\not=\emptyset$.
\end{lemma}
\begin{proof}
Thue Morse words are well known overlap-free ($2^+$-power free) right infinite words on two letters \cite{SHALLIT201996}. Let $x\in\Sigma_k$ and let $t\in\pfl_{k,\alpha}^{\mathbb{N},R}$ be a Thue Morse word on two letters such that $x\not\in\Factor(t)$. Since $k\geq 3$, we have that such $t$ exists.
This ends the proof.
\end{proof}

\section{Bi-infinite $\alpha$-power words}

For the rest of the article suppose that $(k,\alpha)\in\widetilde\Upsilon$; it means $k\geq 3$ and $\alpha\geq \resalpha$.

We define two technical sets $\Gamma$ and $\Delta$.
\begin{definition}
Let $\Gamma$ be a set of triples defined as follows. We have that $(w,\eta ,u)\in\Gamma$ if and only if \begin{itemize}
\item $w\in\Sigma_k^+$, $\eta ,u\in\Sigma_k^*$,  and 
\item if $\vert u\vert\leq \vert w\vert$ then $\vert \eta \vert\geq(\alpha+1)\alpha^{\vert w\vert-\vert u\vert}\vert w\vert\mbox{.}$
\end{itemize}
\end{definition}
\begin{definition}
\label{ccn512rf1x}
Let $\Delta$ be a set of $6$-tuples defined as follows. We have that $(s,\sigma ,w,\eta ,x,u)\in\Delta$ if and only if
\begin{enumerate}
\item $s\in\Sigma_{k}^{\mathbb{N}, L}$, $\sigma ,\eta ,u\in\Sigma_k^*$, $w\in\Sigma_k^+$, $x\in\Sigma_k$, 
\item \label{du77bxn21b} $s\sigma w\eta xu\in\pfl_{k,\alpha}^{\mathbb{N},L}$,
\item \label{du87ejh14} $(w,\eta ,u)\in\Gamma$,
\item \label{ddhy7vzlp} $\occur(s\sigma w,w)=1$, and 
\item $x\not\in\Factor(s)\cup\Factor(u)$.
\end{enumerate}
\end{definition}
Given $(s,\sigma ,w,\eta ,x,u)\in\Delta$ and $y\in\Sigma_k$, let 
\[\begin{split}\Pi(s,\sigma ,w,\eta ,x,u,y)=\{(r,\beta)\mid r\in\Sigma_k^+\mbox{ and }\beta\in\mathbb{Q}\mbox{ and }\beta>\alpha \mbox{ and } \\ uy\in\pfl_{k,\alpha} \mbox{ and }r^{\beta}\in\Suffix(s\sigma w\eta xuy)\}
\mbox{.}\end{split}\]
\begin{remark}
Realize that if $\Pi(s,\sigma ,w,\eta ,x,u,y)\not=\emptyset$ then $s\sigma w\eta xuy\not\in\pfl_{k,\alpha}^{\mathbb{N},L}$. 
\end{remark}
The following lemma is obvious from the definitions of $\Pi$ and $\Delta$. We omit the proof.
\begin{lemma}
\label{rry7idxq2}
If $(s,\sigma ,w,\eta ,x,u)\in\Delta$, $y\in\Sigma_k$, $uy\in\pfl_{k,\alpha}$, $y\not=x$, and $\Pi(s,\sigma ,w,\eta ,x,u,y)=\emptyset$ then $(s,\sigma ,w,\eta ,x,uy)\in\Delta$.
\end{lemma}
For the rest of this section, suppose that $(s,\sigma ,w,\eta ,x,u)\in\Delta$, $y\in\Sigma_k$, $r\in\Sigma_k^+$, and $\beta\in\mathbb{Q}$ are such that $y\not=x$, $\Pi(s,\sigma ,w,\eta ,x,u,y)\not=\emptyset$, and $(r,\beta)\in\Pi(s,\sigma ,w,\eta ,x,u,y)$.
\begin{remark}
Note that the letter $y$ is such that $s\sigma w\eta xuy$ is not $\alpha$-power free, whereas $s\sigma w\eta xu$ is $\alpha$-power free. 
\end{remark}

We show that $xuy$ is a suffix of $r$.
\begin{lemma}
\label{fu78dnbcgd}
We have that 
$xuy\in\Suffix(r)$.
\end{lemma}
\begin{proof}
Realize that if $r\in\Suffix(uy)$ then $x\not\in\Factor(r)$ and consequently $r^{\beta}\in\Suffix(uy)$. This is a contradiction since $uy\in\pfl_{k,\alpha}$. Hence we have that $r\not\in\Suffix(uy)$. It follows that $\vert r\vert>\vert uy\vert$ and hence $xuy\in\Suffix(r)\mbox{.}$
This completes the proof.
\end{proof}
Let $\overline r\in\Sigma_k^*$ be such that $r=\overline rxuy$; Lemma \ref{fu78dnbcgd} asserts that $\overline r$ exists and is uniquely determined. 

We show that if $r^{\beta}$ is shorter than $\eta xuy$ then there is a prefix $\overline \eta $ of $\eta $ such that $s\sigma w\overline \eta xuy$ is $\alpha$-power free.
\begin{proposition}
\label{vv8d9duij33}
If $r^{\beta}\in\Suffix(\eta xuy)$ then 
there is $\overline \eta \in\Prefix(\eta )$ such that \[(s,\sigma ,w,\overline \eta ,x,uy)\in\Delta\mbox{.}\]
\end{proposition}
\begin{proof}
\begin{table}
\centering
\begin{tabular}{|c|c|c|ccccc|}
\hline
$s$ & $\sigma $ & $w$ & \multicolumn{2}{c|}{$\eta $}                                                             & \multicolumn{1}{c|}{$x$} & \multicolumn{2}{c|}{$uy$}  \\ \hline
    &       &     & \multicolumn{1}{c|}{$\overline \eta $}           & \multicolumn{1}{c|}{$xuy\overline r$} & \multicolumn{1}{c|}{}    & \multicolumn{2}{c|}{}  \\ \hline
    &       &     & \multicolumn{1}{c|}{$zr^{\beta-2}\overline r$} & \multicolumn{4}{c|}{$xuyr$}    \\ \hline
    &       &     & \multicolumn{5}{c|}{$zr^{\beta}$}    \\ \hline
\end{tabular}
\caption{Case $\vert r\vert \leq \vert w\vert$ and $\vert u\vert\leq\vert w\vert$}
\label{fvstyd6sd78jh}
\end{table}
Let $z$ be such that $\eta xuy=zr^{\beta}$. Let $\overline \eta =zr^{\beta-2}\overline r$. Then $\sigma w\overline \eta xuy\in\Prefix(\sigma w\eta xu)$ and consequently Property \ref{du77bxn21b} of Definition \ref{ccn512rf1x} implies that \begin{equation}\label{eu7f89fjmmn3}s\sigma w\overline \eta xuy\in\pfl_{k,\alpha}^{\mathbb{N},L}\mbox{.}\end{equation}  
We have that
\begin{align}
\beta\vert r\vert\vert z\vert  &\geq 
(\beta-2)\vert r\vert\vert z\vert \nonumber \\
\implies\quad\beta\vert r\vert\vert z\vert +(\beta-2)\beta\vert r\vert\vert r\vert &\geq 
(\beta-2)\vert r\vert\vert z\vert +(\beta-2)\beta\vert r\vert\vert r\vert \nonumber\\
\implies\quad\left(\vert z\vert +(\beta-2)\vert r\vert\right)\beta\vert r\vert &\geq 
(\vert z\vert +\beta\vert r\vert)(\beta-2)\vert r\vert \nonumber\\
\label{dyuvd888djdu}\implies\quad\frac{\vert z\vert +(\beta-2)\vert r\vert}{\vert z\vert +\beta\vert r\vert} &\geq 
\frac{(\beta-2)\vert r\vert}{\beta\vert r\vert}
\end{align}
It follows from (\ref{dyuvd888djdu}) and $\beta> \alpha\geq \resalpha$ that \begin{equation}\begin{split}\label{ggn78xvcd5}\frac{\vert \overline \eta \vert}{\vert \eta \vert}=\frac{\vert zr^{\beta-2}\overline r \vert}{\vert zr^{\beta}\vert-\vert xuy\vert}\geq\frac{\vert zr^{\beta-2}\vert}{\vert zr^{\beta}\vert}=\frac{\vert z\vert +(\beta-2)\vert r\vert}{\vert z\vert +\beta\vert r\vert}\geq \\ 
\frac{(\beta-2)\vert r\vert}{\beta\vert r\vert}=
 \frac{\beta-2}{\beta}= 1-\frac{2}{\beta}>\frac{3}{5}\mbox{.}\end{split}\end{equation}
From Property \ref{du87ejh14} of Definition \ref{ccn512rf1x},  (\ref{ggn78xvcd5}), and $\frac{3}{5}>\frac{1}{5}\geq\frac{1}{\alpha}$ we have that if $\vert u\vert\leq \vert w\vert$ then \begin{equation}\begin{split}\label{sbvx7778zx3}\vert\overline \eta \vert\geq \frac{3}{5}\vert \eta \vert \geq \frac{1}{\alpha}\vert \eta \vert \geq \frac{1}{\alpha}(\alpha+1)\alpha^{\vert w\vert-\vert u\vert} \vert w\vert> (\alpha+1)\alpha^{\vert w\vert-(\vert u\vert+1)}\vert w\vert\mbox{.}\end{split}\end{equation}
From (\ref{eu7f89fjmmn3}) and (\ref{sbvx7778zx3}) we have that $(w,\eta ,uy)\in\Gamma$ and $(s,\sigma ,w,\overline \eta ,x,uy)\in\Delta$. Realize that Property \ref{du77bxn21b} of Definition \ref{ccn512rf1x} is asserted by (\ref{eu7f89fjmmn3}) and Property \ref{du87ejh14} of Definition \ref{ccn512rf1x} is asserted by (\ref{sbvx7778zx3}). Other properties of Definition \ref{ccn512rf1x} are obvious.
Table \ref{fvstyd6sd78jh} illuminates the structure of the words.

This completes the proof.
\end{proof}

The next lemma shows that if $r$ and $u$ are shorter than $w$ then $r^{\beta}$ is shorter than $\eta xuy$.
\begin{lemma}
\label{ujfdnfg87jdd8v}
If $\vert r\vert \leq \vert w\vert$ and $\vert u\vert\leq\vert w\vert$ then
$r^{\beta}\in\Suffix(\eta xuy)$.
\end{lemma}
\begin{proof}
Property \ref{du77bxn21b} of Definition \ref{ccn512rf1x} implies that \begin{equation}\label{ccnbdh99873jh}\beta\leq \alpha+1\mbox{.}\end{equation}
From $\vert u\vert\leq\vert w\vert$ and Property \ref{du87ejh14} of Definition \ref{ccn512rf1x} it follows that 
 \begin{equation}\label{uid98kjwke9}\vert \eta \vert\geq (\alpha+1)\vert w\vert\mbox{.}\end{equation}  The lemma follows from  (\ref{ccnbdh99873jh}), (\ref{uid98kjwke9}), and  $\vert r\vert \leq \vert w\vert$. 
This ends the proof.
\end{proof}

We show that if $r^{\beta}$ is longer than $\eta xuy$ then $r$ is longer than $w$.
\begin{lemma}
\label{yhdb76sdhv1a09}
If $r^{\beta}\not\in\Suffix(\eta xuy)$ then
$\vert r\vert >\vert w\vert$.
\end{lemma}
\begin{proof}
Lemma \ref{ujfdnfg87jdd8v} asserts that if $\vert u\vert,\vert r\vert \geq \vert w\vert$ then $r^{\beta}\in\Suffix(\eta xuy)$.
Thus it follows from $r^{\beta}\not\in\Suffix(\eta xuy)$ that 
\begin{equation}\label{vbnd788fjen2c}\max\{\vert r\vert,\vert u\vert\}>\vert w\vert\mbox{.}\end{equation}
If $\vert u\vert>\vert w\vert$ then from Lemma \ref{fu78dnbcgd} we have that $\vert r\vert > \vert w\vert$. Then the lemma follows from (\ref{vbnd788fjen2c}).
This completes the proof.
\end{proof}

The next lemma shows that $\eta xuy$ contains $r^{\beta-2}$ as a factor and $w\eta xuy$ contains $r^{\beta-1}$ as a factor.
\begin{lemma}
\label{yv2v2vldihj}
If $r^{\beta}\not\in\Suffix(\eta xuy)$ then \[r^{\beta-2}\in\Suffix(\eta xuy)\quad\mbox{ and }\quad r^{\beta-1}\in\Suffix(w\eta xuy)\mbox{.}\]
\end{lemma}
\begin{proof}
We distinguish two cases:
\begin{itemize}
\item $r^{\beta}\in\Suffix(w\eta xuy)$. Then from Lemma \ref{yhdb76sdhv1a09} it follows that $r^{\beta-1}\in\Suffix(\eta xuy)$ and consequently also  $r^{\beta-2}\in\Suffix(\eta xuy)$.
\item $r^{\beta}\not\in\Suffix(w\eta xuy)$. Then Lemma \ref{yhdb76sdhv1a09} implies that $w\in\Factor(rr)$ and in consequence Property \ref{ddhy7vzlp} of Definition \ref{ccn512rf1x} implies that $r^{\beta-1}\in\Suffix(w\eta xuy)\mbox{.}$
Then from Lemma \ref{yhdb76sdhv1a09} we have that $r^{\beta-2}\in\Suffix(\eta xuy)\mbox{.}$
\end{itemize}
This completes the proof.
\end{proof}

We show that if $r^{\beta}$ is longer than $\eta xuy$ then there is a prefix $\overline \eta $ of $\eta $ such that $s\sigma w\overline \eta xuy$ is $\alpha$-power free.
\begin{proposition}
\label{nj8b3vdjlso}
If $r^{\beta}\not\in\Suffix(\eta xuy)$ then
there is $\overline \eta \in\Prefix(\eta )$ such that \[(s,\sigma ,w,\overline \eta ,x,uy)\in\Delta\mbox{.}\]
\end{proposition}
\begin{proof}
From Lemma \ref{fu78dnbcgd} and Lemma \ref{yv2v2vldihj} we have that \begin{equation}\label{dui8ff9ekj}\vert \eta\vert\geq \vert r^{\beta-3}\vert\mbox{.}\end{equation}
From $r^{\beta}\not\in\Suffix(\eta xuy)$ it follows that \begin{equation}\label{djk8f9erk54d}\vert\eta\vert\leq \vert r^{\beta}\vert\mbox{.}\end{equation}
Let $z\in\Sigma_k^*$ be such that $zr^{\beta-1}=w\eta xuy\mbox{}$ and let $\overline \eta\in\Prefix(\eta)$ be such that $w\overline \eta xuy=zr^{\beta-2}$.
Lemma \ref{yv2v2vldihj} implies that $z,\overline \eta $ exist and are uniquely determined. 
Clearly we have that \begin{equation}\label{hnc87gfhjj8}\vert \overline \eta\vert\geq \vert \eta\vert-\vert r\vert\mbox{.}\end{equation}
Since $\beta>\alpha\geq \resalpha$, it follows from (\ref{dui8ff9ekj}), (\ref{djk8f9erk54d}), and (\ref{hnc87gfhjj8}) we have that
\begin{equation}\label{pce5terdu}\frac{\vert\overline \eta\vert}{\vert\eta\vert}\geq \frac{\vert\overline \eta\vert}{\vert r^{\beta}\vert}\geq\frac{\vert \eta\vert-\vert r\vert}{\vert r^{\beta}\vert}\geq\frac{\vert r^{\beta-3}\vert-\vert r\vert}{\vert r^{\beta}\vert}=\frac{\vert r^{\beta-4}\vert}{\vert r^{\beta}\vert}=\frac{\beta-4}{\beta}>\frac{1}{5}.\end{equation}

From Property \ref{du87ejh14} of Definition \ref{ccn512rf1x},  (\ref{pce5terdu}), and $\frac{1}{5}\geq\frac{1}{\alpha}$ we have that if $\vert u\vert\leq \vert w\vert$ then \begin{equation}\begin{split}\label{dyuejh77d8f}\vert\overline \eta \vert\geq \frac{1}{5}\vert \eta \vert \geq \frac{1}{\alpha}(\alpha+1)\alpha^{\vert w\vert-\vert u\vert} \vert w\vert> (\alpha+1)\alpha^{\vert w\vert-(\vert u\vert+1)}\vert w\vert\mbox{.}\end{split}\end{equation}

It is clear that $w\overline \eta xuy\in\Prefix(w\eta xu)$, hence $s\sigma w\overline \eta xuy\in\pfl_{k,\alpha}^{\mathbb{N},L}$. Then it is easy to verify that $(s,\sigma ,w,\overline \eta ,x,uy)\in\Delta$. Note that Property \ref{du87ejh14} of Definition \ref{ccn512rf1x} follows from (\ref{dyuejh77d8f}).

Table \ref{vmi99sgtxc78} illuminates the structure of the words for this case.
\begin{table}
\centering
\begin{tabular}{|cccccccc|}
\hline
\multicolumn{1}{|c|}{$s$} & \multicolumn{1}{c|}{$\sigma $} & \multicolumn{1}{c|}{$w$} & \multicolumn{2}{c|}{$\eta $}                                                   & \multicolumn{1}{c|}{$x$} & \multicolumn{2}{c|}{$uy$}  \\ \hline
\multicolumn{1}{|c|}{}    & \multicolumn{1}{c|}{}      & \multicolumn{1}{c|}{}    & \multicolumn{1}{c|}{$\overline \eta xuy$} & \multicolumn{1}{c|}{$\overline r$} & \multicolumn{1}{c|}{}    & \multicolumn{2}{c|}{}  \\ \hline
\multicolumn{1}{|c|}{}    & \multicolumn{1}{c|}{}      & \multicolumn{2}{c|}{$zr^{\beta-2}$}                                & \multicolumn{4}{c|}{$r$}   \\ \hline
\multicolumn{1}{|c|}{}    & \multicolumn{1}{c|}{}      & \multicolumn{6}{c|}{$zr^{\beta-1}$}        \\ \hline
\multicolumn{8}{|c|}{$\cdots r^{\beta}$}    \\ \hline
\end{tabular}
\caption{Case $r^{\beta}\not\in\Suffix(\eta xuy)$}
\label{vmi99sgtxc78}
\end{table}
This completes the proof.
\end{proof}

A consequence of the previous two propositions is that in every case there is a prefix $\overline \eta $ of $\eta $ such that $s\sigma w\overline \eta xuy$ is $\alpha$-power free.
\begin{lemma}
\label{dyu77bvqsj}
There is $\overline \eta \in\Prefix(\eta )$ such that $(s,\sigma ,w,\overline \eta ,x,uy)\in\Delta$.
\end{lemma}
\begin{proof}
We distinguish two cases:
\begin{itemize}
\item $r^{\beta}\in\Suffix(\eta xuy)$. In this case the lemma follows from Proposition \ref{vv8d9duij33}.
\item $r^{\beta}\not\in\Suffix(\eta xuy)$ and $r^{\beta}\in\Suffix(w\eta xuy)$. In this case the lemma follows from Proposition \ref{nj8b3vdjlso}.
\end{itemize}
This completes the proof.
\end{proof}

\section{Transition words}

The first theorem of this section shows that if $(s,\sigma ,w,\eta ,x,\epsilon)\in\Delta$ and $t$ is a right infinite $\alpha$-power free word with no occurrence of the letter $x$ then there is a bi-infinite $\alpha$-power free word containing the factor $w$ and having only a finite number of occurrences of $x$.
\begin{theorem}
\label{d78fju5e4}
If $(s,\sigma ,w,\eta ,x,\epsilon)\in\Delta$, $t\in\pfl_{k,\alpha}^{\mathbb{N},R}$, and $x\not\in\Factor(t)$ then there is $\widehat \eta \in\Prefix(\eta )$ such that $s\sigma w\widehat \eta xt\in\pfl_{k,\alpha}^{\mathbb{Z}}$.
\end{theorem}
\begin{proof}
Given $i\in\mathbb{N}\cup\{0\}$, let $\phi(i)\in\Prefix(t)\cap\Sigma_k^i$ be the prefix of $t$ of length $i$.  We have that $\phi(0)=\epsilon$. Let $y_j\in\Sigma_k$ be such that $\phi(j-1)y_j=\phi(j)$.

Let $\omega(0)=\eta$. We have that $(s,\sigma ,w,\omega(0),x,\phi(0))\in\Delta$.
Given $j\in\mathbb{N}$, we define $\omega(j)\in\Prefix(\omega(j-1))$ as follows:

\begin{itemize}
\item
If $\Pi(s,\sigma, w,\omega(j-1),x,\phi(j-1),y_j)=\emptyset$ then we define that $\omega(j)=\omega(j-1)$. Lemma \ref{rry7idxq2} implies that $(s,\sigma ,w,\omega(j),x,\phi(j))\in\Delta$.
\item
If $\Pi(s,\sigma, w,\omega(j-1),x,\phi(j-1),y_j)\not=\emptyset$ then Lemma \ref{dyu77bvqsj} implies that there is $\overline\eta\in\Prefix(\omega(j-1))$ such that $(s,\sigma ,w,\rho,x,\phi(j))\in\Delta$. We define that $\omega(j)=\overline \eta$.
\end{itemize}
It follows that $\omega(j)$ is well defined for all $j\in\mathbb{N}\cup\{0\}$ and also it follows that $(s,\sigma ,w,\omega(j),x,\phi(j))\in\Delta$ for all $j\in\mathbb{N}\cup\{0\}$.

Since $\omega(0)=\eta $ is a finite word and since $\omega(j)\in\Prefix(\omega(j-1))$ for all $j\in\mathbb{N}$, obviously there is $m\in\mathbb{N}$ such that $\omega(i)=\omega(m)$ for all $i\geq m$. Consequently we have that $(s,\sigma ,w,\omega(m),x,\phi(i))\in\Delta$ for all $i\geq m$. Let $\widehat \eta=\omega(m)$. We conclude that $s\sigma w\widehat \eta xt\in\pfl_{k,\alpha}^{\mathbb{Z}}$.

This completes the proof.
\end{proof}

If a factor $w$ is recurrent in a bi-infinite word $v$, then it is possible that $w$ is recurrent only on the ``left or the right side'' of $v$. For the next proposition we will need that the letter $x$ is recurrent ``on the right side''. We define formally this notion. 
\begin{definition}
\label{dxhb7d6djwb}
Suppose $v\in\pfl_{k,\alpha}^{\mathbb{Z}}$ and $w\in\Factor(v)\setminus\{\epsilon\}$. If there are $v_1\in\pfl_{k,\alpha}^{\mathbb{N},L}$ and $v_2\in\pfl_{k,\alpha}^{\mathbb{N},R}$ such that $v=v_1v_2$ and $w\in\Factor_r(v_2)$ then we say that $w$ is \emph{on-right-side recurrent} in $v$.
\end{definition}
\begin{remark}
Note in Definition \ref{dxhb7d6djwb} that no restriction is imposed on the recurrence of $w$ in $v_1$. It means that $w$ may be also recurrent also in $v_1$.
\end{remark}

If $v$ is a bi-infinite $\alpha$-power free word and the letter $x$ is on-right-side recurrent in $v$, then we prove that the main result of the current article holds for $v$.
\begin{proposition}
\label{udi8ds0fne5df}
If $v\in\pfl_{k,\alpha}^{\mathbb{Z}}$, $w\in\Factor(v)\setminus\{\epsilon\}$, $x\in\Factor(w)\cap\Sigma_k$, and $x$ is on-right-side recurrent in $v$ then there is $\widehat v\in\pfl_{k,\alpha}^{\mathbb{Z}}$ such that $w\in\Factor(\widehat v)$ and $x\not\in\Factor_r(\overline v)$.
\end{proposition}
\begin{proof}
It is clear that there are $\widetilde v\in\pfl_{k,\alpha}^{\mathbb{N},L}$, $\ddot v\in\pfl_{k,\alpha}^{\mathbb{N},R}$ and $\widetilde \eta \in\Sigma_k^*$ such that $v=\widetilde vw\widetilde \eta x\ddot v$ and $(w,\widetilde \eta ,\epsilon)\in\Gamma$. Just note that $x$ is on-right-side  recurrent in $v$; hence $\widetilde \eta$ can be chosen to be longer than any arbitrarily chosen positive number.

Let $t$ be a right infinite $\alpha$-power free word on the alphabet $\Sigma_k\setminus\{x\}$. Since $\alpha\geq \resalpha$, Lemma \ref{dy77ejhfiffu} asserts that $t$ exists.
\begin{itemize}
\item If $x\not\in\Factor_r(\widetilde v)$, then let $\overline s=\widetilde v$ and $u=w\widetilde \eta x$.
\item If $x\in\Factor_r(\widetilde v)$, then let $\overline s=t^R$. Clearly $\overline s$ is a left infinite $\alpha$-power free and $x\not\in\Factor(s)$. 
Proposition \ref{nb3rdyd887j} implies that there is a finite word $u$ such that $\overline su\in\pfl_{k,\alpha}^{\mathbb{N},L}$, and $w\widetilde\eta x\in\Suffix(\overline su)$. 
\end{itemize}

Then, obviously, there are $s\in\pfl_{k,\alpha}^{\mathbb{N},L}$ and $\sigma, \eta \in\Sigma_k^*$ such that \[(s,\sigma, w,\eta, x)\in\Delta\quad \mbox{ and }\quad s\sigma w\eta x=\overline su\mbox{.}\] Note that since $x\in\Factor(w)$ and $x\not\in\Factor_r(\overline s)$,  Property \ref{ddhy7vzlp} of Definition \ref{ccn512rf1x} is easy to be asserted by a proper choice of $s$, $\sigma$, and $\eta$.

Theorem \ref{d78fju5e4} implies that there is $\overline \eta$ such that $s\sigma w\overline \eta xt\in\pfl_{k,\alpha}^{\mathbb{Z}}$. Let $\widehat  v=s\sigma w\overline \eta xt$. Clearly $w\in\Factor(\widehat v)$ and $x\not\in\Factor_r(\widehat v)$.
This completes the proof. 
\end{proof}

Now, we can step to the main theorem of the current article.
\begin{theorem}
If $v\in\pfl_{k,\alpha}^{\mathbb{Z}}$, $w\in\Factor(v)\setminus\{\epsilon\}$, then there there are $\overline v\in\pfl_{k,\alpha}^{\mathbb{Z}}$ and $x\in\Sigma_k$ such that $w\in\Factor(\overline v)$ and $x\not\in\Factor_r(\overline v)$.
\end{theorem}
\begin{proof}
\begin{itemize}
\item
If $\Sigma_k\not\subset\Factor_r(v)$ then let $x\in\Sigma_k\setminus\Factor_r(v)$ and let $\overline v=v$.
\item
If $\Sigma_k\subset\Factor_r(v)$ then let $x\in\Factor(w)\cap\Sigma_k$. We have that $x\in\Factor_r(v)$. It follows that either $x$ is on-right-side  recurrent in $v$ or $x$ is on-right-side recurrent in $v^R$. 
\begin{itemize}
\item
If $x$ is on-right-side  recurrent in $v$ then 
 Proposition \ref{udi8ds0fne5df} implies that there is $\widehat v$ such that $w\in\Factor(\widehat v)$ and $x\not\in\Factor_r(\widehat v)$. Let $\overline v=\widehat v$.
\item
If $x$ is on-right-side  recurrent in $v^R$. 
 Proposition \ref{udi8ds0fne5df} implies that there is $\widehat v$ such that $w^R\in\Factor(\widehat v)$. Let $\overline v=\widehat v^R$. Then obviously $w\in\Factor(\overline v)$ and $x\not\in\Factor_r(\overline v)$.
\end{itemize}
\end{itemize}
For every case we showed $\overline v\in\pfl_{k,\alpha}^{\mathbb{Z}}$ and $x\in\Sigma_k$ such that $w\in\Factor(\overline v)$ and $x\not\in\Factor_r(\overline v)$. 
This ends the proof.
\end{proof}

\section*{Acknowledgments}
This work was supported by the Grant Agency of the Czech Technical University in Prague, grant No. SGS20/183/OHK4/3T/14.

\bibliographystyle{siam}
\IfFileExists{biblio.bib}{\bibliography{biblio}}{\bibliography{../!bibliography/biblio}}

\end{document}